%% file: main.tex
\documentclass[a4paper, USenglish, cleveref, autoref, thm-restate, numberwithinsect]{lipics-v2021}

\input{hidelipics.tex} %
\pdfoutput=1 %
\nolinenumbers %

\input{settings.tex}

\title{A Direct-Style Effect Notation\texorpdfstring{\\}{ }for Sequential and Parallel Programs} %
\titlerunning{A Direct-Style Effect Notation for Sequential and Parallel Programs} %

\author%
  {David Richter}%
  {Technical University of Damstadt, Germany}%
  {david.richter@tu-darmstadt.de}%
  {https://orcid.org/0000-0002-8672-0265}%
  {German Federal Ministry of Education and Research \textit{iBlockchain project} (BMBF No.~16KIS0902)}
\author%
  {Timon Böhler}%
  {Technical University of Damstadt, Germany}%
  {timon.boehler@stud.tu-darmstadt.de}%
  {https://orcid.org/0009-0002-9964-7367}%
  {Hessian Ministry of Higher Education, Research, Science and the Arts via the project 3rd Wave of AI (3AI)}
\author%
  {Pascal Weisenburger}%
  {University of St.\texorpdfstring{\,}{ }Gallen, Switzerland}%
  {pascal.weisenburger@unisg.ch}%
  {https://orcid.org/0000-0003-1288-1485}%
  {The University of St.\texorpdfstring{\,}{ }Gallen (IPF, No.~1031569); Swiss National Science Foundation (SNSF, No.~200429)}
\author%
  {Mira Mezini}%
  {Technical University of Damstadt, Germany\texorpdfstring{\\}{; }hessian.AI, Germany}%
  {mezini@informatik.tu-darmstadt.de}%
  {https://orcid.org/0000-0001-6563-7537}%
  {Hessian Ministry of Higher Education, Research, Science and the Arts via the project 3rd Wave of AI (3AI);
   BMBF and the Hessian Ministry of Higher Education, Research, Science and the Arts within their joint support of the \textit{National Research Center for Applied Cybersecurity ATHENE};
   German Federal Ministry of Education and Research \textit{iBlockchain project} (BMBF No.~16KIS0902)}

\authorrunning{Richter, Böhler, Weisenburger, Mezini} %

\Copyright{Jane Open Access and Joan R. Public} %


\category{} %

\supplementdetails[linktext={https://github.com/stg-tud/parseq-notation}]{Software}{https://github.com/stg-tud/parseq-notation} %

\newcommand{\citep}[1]{\cite{#1}}
\newcommand{\citet}[1]{\cite{#1}}

\EventEditors{Karim Ali and Guido Salvaneschi}
\EventNoEds{2}
\EventLongTitle{37th European Conference on Object-Oriented Programming (ECOOP 2023)}
\EventShortTitle{ECOOP 2023}
\EventAcronym{ECOOP}
\EventYear{2023}
\EventDate{July 17--21, 2023}
\EventLocation{Seattle, Washington, United States}
\EventLogo{}
\SeriesVolume{263}
\ArticleNo{40}

\begin{document}

\maketitle

\input{abstract.tex}

\input{introduction.tex}

\input{background.tex}
\input{overview.tex}

\input{intuition.tex}
\input{formalisation.tex}

\input{implementation.tex}

\input{related.tex}

\input{conclusion.tex}

\bibliography{main.bib}

\end{document}

%% file: hidelipics.tex
\hideLIPIcs

%% file: settings.tex
\usepackage{tikz-cd}
\usepackage{listings}
\usepackage{upquote}
\usepackage{booktabs}
\usepackage{adjustbox}
\usepackage{newfloat}

\newcommand{\veryshortarrow}[1][3pt]{\mathrel{%
   \hbox{\rule[\dimexpr\fontdimen22\textfont2-.2pt\relax]{#1}{.4pt}}%
   \mkern-4mu\hbox{\usefont{U}{lasy}{m}{n}\symbol{41}}}}

\usepackage{color}
\input{lstlean.tex}

\DeclareCaptionType{eqnlisting}[Equations][List of Equations]
\crefname{eqnlisting}{Equations}{Equations}
\Crefname{eqnlisting}{Equations}{Equations}

\AtBeginDocument{
  \DeclareCaptionSubType{lstlisting}
  \captionsetup[sublstlisting]{margin=3pt}
  \crefalias{sublstlisting}{listing}
}

\usepackage{bbm}

\usepackage{subcaption}

%% file: lstlean.tex
\definecolor{languagecolor}    {rgb}{0.5, 0.5, 0.5} %
\definecolor{coqkeywordcolor}  {rgb}{0.7, 0.1, 0.1} %
\definecolor{scalakeywordcolor}{rgb}{0.1, 0.1, 0.6} %
\definecolor{tacticcolor}      {rgb}{0.0, 0.1, 0.6} %
\definecolor{commentcolor}     {rgb}{0.4, 0.4, 0.4} %
\definecolor{sortcolor}        {rgb}{0.1, 0.5, 0.1} %
\definecolor{attributecolor}   {rgb}{0.7, 0.1, 0.1} %

\definecolor{symbolcolor}      {rgb}{0.0, 0.1, 0.6} %
\definecolor{nosymbolcolor}    {rgb}{0.0, 0.0, 0.0} %
\newcommand{\symbolcol}{\color{symbolcolor}}

\lstdefinelanguage{coq}  {keywordstyle=[1]{\color{coqkeywordcolor}},
  morekeywords=[1]{
  match, let, forall, Pi, fun, exists, if, then, else,
  do, where, %
  Inductive, CoInductive, Fixpoint, Equations, Class, Reserved, Notation, Infix, Instance,
  }}
\lstdefinelanguage{scala}{keywordstyle=[1]{\color{scalakeywordcolor}},
  morekeywords=[1]{
  import, export, protected, private, public, override, infix, extension,
  val, var, def, class, object, package, trait,
  given, using, with, extends, deriving, implicit, summon,
  match, case,
  if, then, else, break, continue, return, try, catch, for, yield, do,
  macro,
  purify,
  }}

\newcommand{\languagetag}{}

\makeatletter
\lst@AddToHook{DeInit}{\languagetag}
\makeatother

\lstnewenvironment{coqlisting}[1][]
  { \renewcommand{\languagetag}{\parbox[c][0pt]{\linewidth}{\raggedleft\scriptsize\color{languagecolor} Coq}}%
    \lstset{language=coq, #1}}{}
\lstnewenvironment{scalalisting}[1][]
  { \renewcommand{\languagetag}{\parbox[c][0pt]{\linewidth}{\raggedleft\scriptsize\color{languagecolor} Scala}}%
    \lstset{language=scala, #1}}{}
\lstnewenvironment{myequations}[1][]
  {\renewcommand{\symbolcol}{\color{nosymbolcolor}}\noindent\adjustbox{center,vspace=\bigskipamount}\bgroup\lstset{#1, frame=none}}
  {\egroup\renewcommand{\symbolcol}{\color{symbolcolor}}}

\lstMakeShortInline[columns={[l]fullflexible},language=coq]|
\lstMakeShortInline[columns={[l]fullflexible},language=scala]^

%% file: abstract.tex
\begin{abstract}
Modeling sequential and parallel composition of effectful computations
has been investigated in a variety of languages for a long time.
In particular, the popular do-notation provides a lightweight effect embedding for any instance of a monad.
Idiom bracket notation, on the other hand, provides an embedding for applicatives.
First, while monads force effects to be executed sequentially, ignoring potential for parallelism, applicatives do not support sequential effects.
Composing sequential with parallel effects remains an open problem.
This is even more of an issue as real programs consist of a combination of both sequential and parallel segments.
Second, common notations do not support invoking effects in direct-style,
instead forcing a rigid structure upon the code.

In this paper, we propose a mixed applicative/monadic notation
that retains parallelism where possible,
but allows sequentiality where necessary.
We leverage a direct-style notation where sequentiality or parallelism is derived from the structure of the code.
We provide a mechanisation of our effectful language in Coq and prove that our compilation approach retains the parallelism of the source program.
\end{abstract}

%% file: introduction.tex
\section{Introduction}

Programming language designers often select a few common effects (state, IO, network) and bake them into the language.
It is, however, impossible to predict what effects developers will need in the future
(as was the case with integrated queries~\cite{MeijerBB06, Meijer07, BiermanMT07},
reactive programming~\cite{Meijer12},
asynchronous programming~\cite{SymePL11,BiermanRMMT12},
multitier programming~\cite{NeubauerT05,RichterKWSFM22},
differentiable programming \cite{Karczmarczuk98}, etc).
Thus, we argue that language designs should be equipped with support for developer-implementable effects.

Modeling effectful computations has long been the subject of investigation
in the context of various languages.
As a result, there exist different abstractions and notations with different properties.
A prominent abstraction for modeling effectful computations 
are monads (e.g., known from Haskell) and the do-notation that emerged from monadic comprehensions~\citep{Wadler92}
providing a lightweight way to embed monads into programs.
But monads force effects to be executed sequentially,
ignoring potential for parallelism.
Notations for parallelism,
such as idiom bracket notation for applicatives~\citep{McbrideP08}, on the other hand,
do not support sequential effects.
Yet, programs are rarely only parallel or only sequential;
thus it is desirable to support both sequential and parallel composition of effectful operations.

To the best of our knowledge, there are no approaches that provide such support.
The |ApplicativeDo| approach by Marlow et al.\@ \cite{MarlowBCP14, MarlowJKM16}
attempts to retrofit parallelism into the do-notation, i.e.,
with |ApplicativeDo|, developers write code using the do-notation and
an optimising compiler tries to infer which computations are parallelizable.
Yet, in the general case, it is not possible to decide statically
whether two monadic operations are actually parallelizable or whether the result of one operation depends on the execution of the other.
Hence, there is a danger that the compiler either incorrectly decides
that two operations can be executed in parallel,
which can lead to race conditions, or conservatively decides
to not parallelize operations that could actually be parallelized,
reducing the potential for improved performance.
To counteract race conditions, the |ApplicativeDo| approach requires developers to adhere to specific coding conventions such as
only using expressions which are either all read-only or write-only~\cite{MarlowBCP14}.

Another weak point of Haskell's do-notation (and thus also of the approach by Marlow et al.)
is that it enforces a specific structure upon the code 
with strict adherence to one effect per line, which
does not allow effects in arbitrary places.
The do-notations in Idris \cite{idris-bang-notation} and Lean~\cite{lean-do-notation, UllrichM22} are less restrictive and support direct-style effect usage.
Scala supports both structuring effectful code in do-notation via for-comprehensions
and for in some cases in direct-style via async-await~\cite{scala-async};
but, both are based on monads, thus force sequential execution.
Although async-await was explicitly designed for concurrency,
developers must be careful to start parallel execution
before accessing their result
to preserve parallelism.

In this paper, we propose a direct-style notation that enables 
sequential and parallel composition of effectful operations
(using monads and applicatives, respectively) without forcing a specific structure of the code.
We present a one-pass translation from direct-style to monadic effect combinators.
Instead of trying to infer the potential for parallelism on top of sequential programs, 
our approach preserves parallelism that is inherent in the structure of the code thanks to direct style.
This makes it easier to reason about the correctness of the proposed translation process
and we present a correctness proof and its mechanization in Coq~\cite{coq-8-16}.
We conceptualize the preservation of parallelism as the span of a term
(the length of the longest path of effectful operations)
and the work of a term (the sum of all effectful operations therein).
Our translation is span-preserving leveraging applicatives and monads.
In contrast, notations based on monads alone are not span-preserving,
as they have to chain all effectful operations into a single sequence.

Our compilation has an elegant description
as a set of equations forming a structurally recursive function over the syntax,
whose equations are the well-known monadic and applicative laws and free theorems.
Implementations for do-notation are essentially compilers for an effectful language.
They can produce efficient code by avoiding administrative redexes
and generating proper tail calls.
Including this optimisation in standard effectful languages modeled by monads
can be seen as performing partial evaluation of the code
via the semantics, extended by the monad laws.
In our case, we target mixed monadic and applicative code.
Therefore, our optimised translation also combines the use of the monadic laws with the use of applicative laws.

\subparagraph{Contributions.}
In summary, this paper makes the following contributions:
\begin{itemize}
\item We present the first mixed applicative/monadic direct-style effect notation.
\item We formalize an optimised one-pass translation from direct-style to effect combinators.
\item We prove that our translation preserves typability, semantics, and parallelism.
\item We mechanize the proof using parametric higher order abstract syntax.
\item We implement the proposed translation in the Scala programming language using Scala macros,
      which enables us to stay close to the formal development.
\end{itemize}

\subparagraph{Structure.}
The remainder of the paper is structured as follows.
\Cref{sec:overview} provides code examples and an intuitive overview of our approach.
\Cref{sec:formalisation} formally defines the proposed translation and provides a proof that it preserves typability, semantics, and parallelism, which is mechanized using parametric higher order abstract syntax.
\Cref{sec:implementation} presents the implementation in Scala.
\Cref{sec:related-work} surveys related work.
\Cref{sec:conclusion} concludes the paper and presents ideas for future work.

%% file: overview.tex
\section{Overview}
\label{sec:overview}

In this section, we (a) briefly discuss the difference between monadic,
applicative, and mixed notations by examples in Scala, and (b)  
informally present our mixed direct notation and its implementation 
by translation to effect combinators.

\subsection{Monadic, Applicative, Mixed and Direct Style Notations}

\subparagraph{Functors, Applicatives and  Monads.}
A functor (in functional programming) for |F: Type → Type| is a method |map|
that turns a function on values into a function on values wrapped in the functor.
Intuitively, a value of type |F A| represents an effectful computation of type |A|.
An applicative for |F| is a functor and a method |pure| to wrap a value into the functor,
and a method |ap| (which we occasionally also write |f ⋄ x| instead of |ap f x|)
to apply an effectful computation returning a function,
to an effectful computation returning an argument.
A monad for |F| is an applicative for |F| and a method |bind|,
which runs an effectful computation and feeds the resulting value to another effectful computation.
Below we show the mathematical description and an encoding in Scala via traits:

\medskip\noindent\begin{minipage}{.5\textwidth}
\begin{myequations}
map:  (A → B) → (F A → F B)
pure: A → F A
ap:   F (A → B) → (F A → F B)
bind: (A → F B) → (F A → F B)
\end{myequations}
\end{minipage}
\begin{minipage}{.5\linewidth}
\begin{scalalisting}[columns={[l]fullflexible}]
trait Functor[F[_]]:
  def map(f: A => B, a: F[A]): F[B]
trait Applicative[F[_]] extends Functor[F]:
  def pure(a: A): F[A]
  def ap(f: F[A => B], a: F[A]): F[B]
trait Monad[F[_]] extends Applicative[F]:
  def bind(f: A => F[B], a: F[A]): F[B]
\end{scalalisting}
\end{minipage}

For convenience, we will write ^pure(x)^, ^x.bind(f)^ and ^f.ap(x)^,
so we define them in Scala as a ^extension^ methods for every object which has a corresponding instance,
and |pure(x)| as a top-level function.
Note that we swapped arguments for bind as a method ^x.bind(f)^
with regard to its type as a function ^bind(f)(x)^.

\subparagraph{Monadic Notation.}
To illustrate monadic notations,
consider the two lines of code below (left side) that use a for-comprehension ^for ... yield ...^.\footnote{
  For-comprehensions \lstinline[language=scala]|for ... yield ...|
  are Scala's equivalent of Haskell's do-notation \lstinline[language=scala]|do ...; return ...|
  The main difference besides Scala's and Haskell's monadic notation is
  that every for-comprehension must end with a \lstinline[language=scala]|yield|.
  Yet, this does not reduce expressive power,
  as any do-block without a final return can be expressed with an additional binding,
  e.g., \lstinline[language=scala]|do ...; e| can be represented by \lstinline[language=scala]|for ...; tmp ← e  yield tmp|.}
The programs execute two effectful statements ^fetchX^ and ^fetchY^ and bind the result in variables ^x^ and ^y^, respectively,
to be combined by a function call to |f|.
The for-comprehension (monadic notation) is desugared into explicit use of monadic |bind| (right side).

\medskip\noindent\begin{minipage}{.45\textwidth}
\begin{scalalisting}[columns={[l]fullflexible}]
for x ← fetchX; y ← fetchY  yield f(x)(y)
for y ← fetchY; x ← fetchX  yield f(x)(y)
\end{scalalisting}
\end{minipage}
\begin{minipage}{.55\textwidth}
\begin{scalalisting}[columns={[l]fullflexible}]
fetchX.bind(x => fetchY.bind(y => pure(f(x)(y))))
fetchY.bind(y => fetchX.bind(x => pure(f(x)(y))))
\end{scalalisting}
\end{minipage}

\subparagraph{Applicative Notation.}
In the programs above, ^x^ does not depend on ^y^ and vice versa.
If all statements in the program of an applicative are independent from each other --
e.g., none of the variables that are introduced in the for-comprehension are used in the ^for^ part,
but only after the ^yield^, which has access to all variables introduced above --
we can interpret the for-comprehension as an applicative notation instead of monadic notation.
Then, the program below on the left side would be translated into the program below on the right side,
using applicative ^ap^ to encode actual parallelism,
where ^ap^ is parallel execution followed by function application.
In the example program, ^pure(f)^, ^x^ and ^y^ are executed in parallel,
followed by the function application of the result of ^pure(f)^
to the result of ^x^ and the result of ^y^:

\medskip\noindent\begin{minipage}{.45\textwidth}
\begin{scalalisting}[columns={[l]fullflexible}]
for x ← fetchX; y ← fetchY  yield f(x)(y)
\end{scalalisting}
\end{minipage}
\begin{minipage}{.55\textwidth}
\begin{scalalisting}[columns={[l]fullflexible}]
pure(f).ap(fetchX).ap(fetchY)
\end{scalalisting}
\end{minipage}

\subparagraph{Mixed Notation.}
To illustrate where these notations for effectful computations fall short,
consider the following larger program
that fetches four resources from the Internet.
The program first fetches a resource ^urlXX^, which contains another url ^urlX^,
and then fetches a resource from ^urlX^ and stores it in ^x^.
The program further fetches a resource ^urlYY^, which contains another url ^urlY^,
and then fetches the resource from ^urlY^ and stores it in ^y^.
Finally, the program concatenates ^x^ and ^y^:

\medskip\noindent\begin{minipage}{.475\textwidth}
\begin{scalalisting}
val urlXX = "https://example.org/configx"
val urlYY = "https://example.org/configy"
for urlX  ← fetch(urlXX)
    x     ← fetch(urlX)
    urlY  ← fetch(urlYY)
    y     ← fetch(urlY)
    yield x ++ y
\end{scalalisting}
\end{minipage}%
\hspace{.05\textwidth}%
\begin{minipage}{.475\textwidth}
\small
\begin{tikzcd}
                                          & \bullet \arrow[ld] \arrow[rd] &                                           \\
\text{\lstinline|fetch(urlXX)|} \arrow[d] &                               & \text{\lstinline|fetch(urlYY)|} \arrow[d] \\
\text{\lstinline|fetch(urlX)|} \arrow[rd] &                               & \text{\lstinline|fetch(urlY)|} \arrow[ld] \\
                                          & \text{\lstinline|x ++ y|}     &                                           
\end{tikzcd}
\end{minipage}

\bigskip

Observe that ^urlXX^ needs to be fetched before ^urlX^ can be fetched and
^urlYY^ needs to be fetched before ^urlY^ can be fetched. But
fetching ^urlXX^ and ^urlYY^ is independent from each other;
and so is ^urlX^ and ^urlY^ (as illustrated in the diagram above).
Thus, the example contains both parallel and sequential elements.
However, if implemented via monadic notation,
the program will run sequentially, fetching ^urlXX^, then ^urlX^, then ^urlYY^, then ^urlY^.
The applicative notation, on the other hand, is not even possible
because it would require all effects to be independent from each other.

\subparagraph{Direct-Style Mixed Notation.}

In our direct-style notation, we combine the syntactic form ^for ... yield ...^
into a single instruction ^purify^.
Now, the program can be written to look like the following snippet, which reads
\emph{Concatenate
(1)~the result of fetching the value pointed to by a URL by fetching} |urlXX| \emph{with
(2)~the result of fetching the value pointed to by a URL by fetching} |urlYY|:
\begin{scalalisting}[columns={[l]fullflexible}]
purify:
  fetch(fetch(urlXX).↓).↓ ++ fetch(fetch(urlYY).↓).↓
\end{scalalisting}

The ^purify^ operation introduces an operation of type ^↓: F[X] => X^ used like ^... . ↓^ into the local scope,
which represents direct-style effect execution.
If the enclosed code does not make use of the ^↓^, the operation ^purify^ works exactly like ^pure^.
Otherwise, effect execution ^↓^ is translated into proper use of ^bind^ and ^ap^.

In direct style, potential for parallelism is implicitly defined by the structure of the code.
In particular, the function arguments of concat ^++^
naturally have no dependency on each other, and can therefore be executed in parallel.
The corresponding program with explicit mixed monadic/applicative combinators is:
\begin{scalalisting}[columns={[l]fullflexible}]
pure(x => y => x ++ y).ap( fetch(urlXX).bind(fetch) ).ap( fetch(urlYY).bind(fetch) )
\end{scalalisting}

The approach we propose in this paper exploits the information encoded in our direct-style notation 
to define a compositional (e.g., structurally recursive) and provably correct 
compiler that transforms such direct-style programs into semantically 
equivalent mixed applicative/monadic programs.

\subsection{Discussion}

We discuss the similarities and differences between different notations.
For illustration, consider that
monadic code in the for/yield notation can be easily refactored into direct style,
roughly by replacing ^←^ with ^= ↓^.
In turn, monadic code is compiled into explicit use of monadic operators
by calling bind on each value and calling map on the last:

\qquad

\noindent\begin{minipage}{.33\textwidth}
\begin{scalalisting}
for
  x ← a
  y ← b
  z ← c
  yield e
\end{scalalisting}
\end{minipage}
\begin{minipage}{.33\textwidth}
\begin{scalalisting}
purify:
  val x = a.↓
  val y = b.↓
  val z = c.↓
  e
\end{scalalisting}
\end{minipage}
\begin{minipage}{.33\textwidth}
\begin{scalalisting}

a.bind { x =>
  b.bind { y =>
    c.map  { z =>
      e }}}
\end{scalalisting}
\end{minipage}

\subparagraph{Scaling.}
A benefit of direct-style code is that it ``scales'' better
for larger programs in the sense that it integrates well with common language constructs.
In particular, direct style composes better with branching.
Consider the following versions of the same program.
On the left, the program is written in monadic notation, implemented in pure Scala,
which requires us to leave and enter monadic notation a second time.
The program fetches the time of the last change and the last caching of a certain request.
If the resource has been updated since the last caching,
we count the cache miss for statistics,
request the resource freshly, pass it to the cache and return it.
Otherwise, we count the cache hit for statistics, and return the answer from the cache.
Now, compare the program on the left with a clearer direct-style representation on the right, which is implemented 
using our approach (code on the right).

\medskip\noindent\begin{minipage}{.48\textwidth}
\begin{scalalisting}
for freshtime <- fetch(freshtimeUrl)
    cachetime <- fetch(cachetimeUrl)
    result    <-
      if freshtime >   cachetime
      then for _    <- countFresh
               tmp1 <- fetchFresh
               tmp2 <- writeCache(tmp1)
               yield tmp2
      else for _    <- countCache
               tmp  <- readCache
               yield tmp
    yield result
\end{scalalisting}
\end{minipage}
\begin{minipage}{.52\textwidth}
\begin{scalalisting}

purify:
  if fetch(freshtimeUrl).↓ >
     fetch(cachetimeUrl).↓
  then
    countFresh.↓
    parseAndCache(fetchFresh.↓).↓
  else
    countCache.↓
    readCache.↓
~~
\end{scalalisting}
\end{minipage}

\begin{table}[b]
\caption{Subnotations.}\label{fig:subnotations}
\vspace{-1em}
\begin{center}\begin{tabular}{lll}
\toprule
Scheme                             & Description     & Typeclass\\
\midrule
^purify{ ... ↓ ... }^                    & one mark          & Functor     ^map^  \\
^purify{ ... ↓ ... ↓ ... }^              & multiple marks    & Applicative ^ap^   \\
^purify{ ... (... ...↓ ...).↓ ... }^     & nested marks      & Monad       ^bind^ \\
^purify{ ...; val x = ...↓; ... ↓ ... }^ & consecutive marks & Monad       ^bind^ \\
\bottomrule
\end{tabular}\end{center}
\end{table}

\subparagraph{Sub-notations.}
Direct-style notation subsumes three different explicit effect notations (\cref{fig:subnotations}).
First, if a ^purify^ expression contains exactly one down arrow ^↓^ as a mark for effect execution (Row~1),
then the notation translates to solely using the Functor interface.
This case corresponds to a standard map operation.

Second, if a ^purify^ expression contains multiple such marks, which are ``parallel'' with regard to each other (Row~2)
-- i.e., when they are side-by-side inside different arguments to a function --
then the expression translates to solely using the Applicative interface.
Crucially, in this case, different effect executions cannot depend on each other.
We call these effect executions ``parallel'' as contrasted with ``sequential'' code, where a statement can depend on the previous one.
Such parallel composition enables parallel execution of effects at run time.

Third, if the expression makes use of nested marks (Row~3)
-- or equivalently of marks which make use of previously bound variables which contained marks (Row~4) --
then the expression translates to using the full Monad interface,
which models sequential code.

Direct-style enables to define parallelism naturally by structuring code such that the execution of effects are independent,
which gives rise to parallel execution of code where this is possible and using sequential execution where necessary.

%% file: intuition.tex
\subsection{From Laws to a Rewrite System}

We refresh the laws of functors, applicatives, and monads and
give an intuition how they can be used to optimise effectful programs.
Then we state a completion of the laws into a rewrite system.
We use the symbol |·| for function composition as in |f · g|,
and use the symbol |(·)| as the name of function composition,
when not used as an infix operator, i.e., |(·) f g v := f (g v)|.

\subparagraph{Laws.}
Typically, the coherence laws are formulated as follows~\cite{McbrideP08, LindleyWY11}.

For Functors,
|map| preserves identity and composition, i.e.,
applying the identity function to an effectful computation
is the same as not doing anything; and applying
two function in sequence to an effectful value
is the same as applying the composite once.
\begin{myequations}
identity      : map id v          =  id v
composition   : map (f · g) v     =  map f (map g v)
\end{myequations}

For Applicatives,
|pure| creates a effect-free, i.e., \emph{pure} value from a value.
The \emph{homomorphism law} states that,
applying a pure function to a pure argument is pure.
The \emph{identity law} states that, if the function is just pure, then there is nothing to do.
The \emph{interchange law} states that, if the argument is pure,
then we can swap the pure argument with the effectful function.
\begin{myequations}
homomorphism  : pure f ⋄ pure v   =  pure (f v)
identity      : pure id ⋄ v       =  v
interchange   : f ⋄ pure v        =  pure (λ f', f' v) ⋄ f
composition   : u ⋄ (v ⋄ w)       =  pure (·) ⋄ u ⋄ v ⋄ w
\end{myequations}

For Monads,
the first and second laws state that executing a pure computation
amounts to not having to execute any effect at all,
allowing us to eliminate the |bind|.
The third law states that |bind| is associative, i.e.,
applying two effectful functions in sequence is the same as applying
the effectful composite of the two functions once.
\begin{myequations}
leftunit      : bind f (pure v)   =  f v
rightunit     : bind pure v       =  v
associativity : bind g (bind f v) =  bind (bind g · f) v
\end{myequations}

\subparagraph{Free theorems.}
The laws are phrased as an equational theory -- to create a compiler from the laws,
we need to rephrase them as a terminating rewrite system.
To do so, we first complete our set of equations with the following free theorems.
Free theorems hold in programming languages with parametric polymorphism
  by parametricity for free~\cite{Wadler89, Voigtlander19},
therefore they are often not stated specifically in the laws,
because there is no additional effort required to make them hold.
On the other hand, as we are interested in making use of laws for optimisation purposes,
we are allowed to make use of the free theorems as well.

Consider the function |pure: ∀ A, A → F A|.
Because it must work over all |A| it cannot change or create new elements of type |A|,
but only duplicate or forget values of type |A|.
Therefore, it does not matter
whether we apply a function |g| to change the |A|s into |B|s before or after applying |pure|.
On the left we apply |f| on the argument of |pure|, on the right we apply |f| on the result:
\begin{myequations}
free_pure:  map f (pure v)  =  pure (f v)
\end{myequations}

Similarly, consider the function |ap: ∀ A B, F (A → B) → F A → F B| and |bind: ∀ A B, (A → F B) → (F A → F B)|.
On the left we apply |f| on the argument of |ap| and |bind|, on the right we apply |f| on the result,
  where \lstinline!(f·)! stands for \lstinline!λ g, f · g!:
\begin{myequations}
free_ap   : ap (map (f·) g) v   =  map f (ap g v)
free_bind : bind (map f · g) v  =  map f (bind g v)
\end{myequations}

We instantiate |g| with |pure id| in the applicative case and with |pure| in the monadic case,
then we can extend the equation chain to the left by the free theorem of pure (|map f · pure = pure · f|),
and to the right by the identity applicative law (|ap (pure id) v  =  v|)
respectively the left-unit monad law (|bind pure v = v|):
\begin{myequations}
*free_ap  : ap (pure f) v     = ap (map (f·) (pure id)) v = map f (ap (pure id) v) = map f v
*free_bind: bind (pure · f) v = bind (map f · pure) v     = map f (bind pure v)    = map f v
\end{myequations}

In fact, these two equation share a common right-hand side,
and thus we can combine them to get a connection between applicative ap and monadic bind:
\begin{myequations}
ap_bind:   ap (pure f) v =  bind (pure · f) v
\end{myequations}

\subparagraph{Completion.}
We can use the free theorems to complete the functor, applicative and monad laws
into a more suitable form.
In particular, we replace
the identity law of the applicative
with their generalization derived above.
Similarily, the right hand side of the interchange law contained the left hand side of the identity law,
therefore we simplify it by composition with the identity law.
Further, observe that the homomorphism law becomes superflous,
as it can be constructed by applying the identity law (or equivalently by the interchange law)
followed by the free theorem of pure;
however we will still make use of it in swapped direction,
such that reading the laws from left to righ,
it does not overlap with the other applicative laws.
The complete set of equations is now:
\begin{myequations}
identity      : map id v           =  v
composition   : map f (map g v)    =  map (f · g) v

homomorphism  : pure (f v)         =  pure f ⋄ pure v             -- swapped
identity      : pure f ⋄ v         =  bind (λ v', pure (f v')) v  -- generalised by ap_bind
interchange   : f ⋄ pure v         =  bind (λ f', pure (f' v)) f  -- combined with identity
composition   : u ⋄ (v ⋄ w)        =  map (·) u ⋄ v ⋄ w           -- combined with identity

leftunit      : bind f (pure v)    =  f v
rightunit     : bind pure v        =  v
associativity : bind g (bind f v)  =  bind (bind g · f) v
\end{myequations}

Looking at the equations above,
we see that the identity and interchange law show that a non-pure argument to ap on the left and on the right
can each be represented by a bind, so one might think
both laws can be unified by a single law, using two binds like
|fs ⋄ xs = bind (λ f, bind (λ x, pure (f x)) xs) fs|.
However, it is not valid to assume this equation.
Actually, there are \emph{at least two} possible trivial instances of an applicative for any monad,
the left-to-right applicative above, but also the right-to-left applicative:
|fs ⋄ xs = bind (λ x, bind (λ f, pure (f x)) fs) xs|.
There is no reason to prefer one over the other,
and, in general, the assumption of either of these equations is too strong;
committing to one such equation would allow elimination of all |ap|s into |bind|s,
and thus implies full sequentiality.
To support parallelism, we have to make neither assumption
and only rely on the equations derived from the applicative laws.

\subsection{Translation}
\label{sec:translation}
We present a rewrite system based on the laws,
and prove its terminating by phrasing it as a structurally recursive function.

We distinguish between a source language and a target language below. %
The source language consists of term formers for variables, function application,
and the direct-style effect execution |↓| represented as |Each|.
The target language consists of term formers for variables, function application, and effect combinators |Pure|, |Bind| and |Ap|;
and parallelism is explicitly structured by those combinators.

\begin{myequations}
(Src)  e, f ::= Var x | App f e | Each e
(Tgt)  g, h ::= Var x | App g h | Pure g | Ap g h | Bind g h
\end{myequations}

The source language uses direct style.
In programs written in the source language, parallelism is implicitly defined by the structure of the code.
In particular, function arguments naturally have no dependency on each other, and can therefore be executed in parallel.
Our compiler translates direct style into monadic and applicative combinators.
The essence of our compilation strategy is to use the monadic and applicative laws
directly as the actual transformation rules.

\subparagraph{Basic Translation.}
The translation starts with the |PURE| expression,
which is implemented as a structurally recursive function over syntax,
expanding the direct-style use of effects |←| into the effect operation |Bind|,
while variables are wrapped in a |PURE|,
and function application is translated to applicative |Ap|,
realising that function arguments can be executed in parallel.

\begin{myequations}
PURE: Src → Tgt
PURE (Var x)    =  Pure (Var x)
PURE (Each e)   =  Bind (PURE e) id
PURE (App f e)  =  Ap (PURE f) (PURE e)
\end{myequations}

\subparagraph{Optimising Translation.}
If we only cared about a correct translation
from the direct-style notation
to the pure calculus with explicit combinators,
then the translation we discussed so far is sufficient.
Yet, we consider an optimising translation (\cref{fig:opt-translation}),
where instead of term using the constructors |Bind| and |Ap| (capitalized) directly,
we use the smart constructors |AP| and |BIND| (all capitals) instead.
Both the constructor and the smart constructor of a term do construct terms
that are semantically indistinguishable, i.e., |AP f x| $\approx$ |Ap f x|.
Smart constructors, however, internalize the optimisation
by reducing to a simpler term if possible.
The translation |PURE| we have described earlier
can be seen as such a smart constructor for the |Pure| term constructor.
It also preserves the semantics, i.e., |PURE x| $\approx$ |Pure x|.

The only difference between the basic and the optimised translation,
is that the optimised smart constructor |PURE| calls to the smart constructor |AP|
instead of using the term constructor |Ap| directly, which can lead to further optimisations.
In this way, we can leverage smart constructors
to integrate the translation with an optimisation
into a one-pass optimised translation.
For the optimisation, the smart constructors apply the monadic and applicative laws,
only in the other direction than the translation, i.e.,
bubbling up |Pure| in a structurally recursive way through the term,
and thereby removing superflous effect combinators in the generated code.

\input{fig-opt-translation.tex}

In particular, |AP|~(\cref{fig:opt-translation}) will reduce the applicative application of a pure function to a pure argument back into the pure function application with only the result wrapped into |Pure| (which is simply the reverse rule of the homomorphism law we used above). Similarly, if either side of |AP| is pure, there are no two effects to be executed in parallel but just a single effect. Hence, we can reduce the term to a single monadic bind.
Finally, if neither argument to |AP| is marked as pure, then we simply return the actual term former |Ap| 
and retain the parallelism.
The optimisation rules that apply to |BIND|~(\cref{fig:opt-translation}) are similar.
If either of its arguments is marked as pure, we can avoid performing effects at all.
If we have nested binds, we can apply the associativity rule to generate a chain of binds.

Overall, seven of the ten equations above come from our generalized laws;
two hold by semantic \emph{indistinguish}ability,
and one is the basis for our \emph{effect translation},
namely the translation of the imperative |←| to an explicit |bind|.

In the following section,
we extend the language,
formalize the language and the translation using a Coq mechanisation,
and prove correctness.

%% file: fig-opt-translation.tex
\begin{figure}
\begin{minipage}{.5\linewidth}
\begin{myequations}
(Src)  e,f ::= Var x | App f e | Each e
(Tgt)  g,h ::= Var x | App g h
             | Pure e | Ap g h | Bind g h
\end{myequations}
\end{minipage}
\begin{minipage}{.5\linewidth}
\begin{myequations}

PURE:  Src → Tgt
AP:    Tgt → Tgt → Tgt
BIND:  Tgt → Tgt → Tgt
~~
\end{myequations}
\end{minipage}
\vspace{-10pt}
\begin{myequations}
PURE (Var x)      =  Pure (Var x)                   -- indistinguishable
PURE (Each e)     =  BIND id (PURE e)               -- effect translation
PURE (App f e)    =  AP (PURE f) (PURE e)           -- homomorphism law

AP (Pure f) e     =  BIND (λ x, Pure (App f x)) e   -- identity law
AP f (Pure e)     =  BIND (λ x, Pure (App x e)) f   -- interchange law
AP f e            =  Ap f e                         -- indistinguishable

BIND g (Bind f e) =  BIND (Bind g · App f)) e       -- associativity law
BIND f (Pure e)   =  App f e                        -- left unit law
BIND Pure e       =  e                              -- right unit law
\end{myequations}
\caption{Optimised Translation.}\label{fig:opt-translation}
\end{figure}

%% file: formalisation.tex
\section{Mechanisation}
\label{sec:formalisation}

We define the source language that features our effect notation and a translation
to a target language which is the subset of the source language that does not include the effect notation.
We prove that our translation preserves typability,
semantics, and parallel execution, which we measure through the program's span and work.
We have mechanized our language and proofs in Coq.

\subsection{Definitions}

We use (parametric) higher-order abstract syntax (PHOAS)~\cite{PfenningE88, Chlipala08},
which enables us to reuse the binders of the host language as binders of the guest language.
PHOAS avoids the need to define first-order syntax, an operational semantics and capture-avoiding substitution, thereby removing intricate lemmas regarding substitution and hundreds of lines of code from the mechanisation, bringing the proof more in line with a more legible pen-and-paper formalisation.

Further, we use intrinsically typed terms~\cite{Danielsson06, BentonHKM12, AltenkirchCU14, AltenkirchCU14b},
and a type-theoretic semantics~\cite{HarperS00}.
Using intrinsically typed terms together with dependent pattern matching allows us to define total evaluation
(in contrast to using untyped terms or simple pattern matching where we could just define partial evaluation).
The reason is that such an approach only needs to consider well-formed terms that \emph{don't go wrong}~\cite{Milner78}.

The common strategy behind all these approaches is to carve out a subset of the host language, that is the language we want to define (the guest language), and then reusing all the power of the host language to define the guest language, avoiding having to reimplement tedious implementation details:
The guest types simply mirror the host types, the guest terms mirror the host terms, and the evaluation function maps guest terms to host terms.

\begin{coqlisting}[float, caption={Lawful Monad.}, label={fig:lawfulmonad}]
Class Monad F := {
  map  {A B}: (A  →  B) → F A → F B;
  pure {A}  : A → F A;
  ap   {A B}: F (A → B) → F A → F B;
  bind {A B}: (A → F B) → F A → F B;
}.

Class LawfulMonad F := {
  monad :> Monad F;

  idl {A B} (f: A → F B)                {x}:  bind f (pure x)    =  f x;
  idr {A B} (f: A → B)                  {x}:  bind (pure · f) x  =  map f x;
  asc {A B C} (f: A → F B) (g: B → F C) {x}:  bind g (bind f x)  =  bind (bind g · f) x;

  apl  {A B} (f: A → B)                 {x}:  ap (pure f) x      =  map (λ x', f x') x;
  apr  {A B} (f: F (A → B))             {x}:  ap f (pure x)      =  map (λ f', f' x) f;
  aplr {A B} (f: A → B)                 {x}:  map f (pure x)     =  pure (f x);

  map_map {A B C} (g: A → B) (f: B → C) {x}:  map f (map g x)    =  map (λ x, f (g x)) x;
}.
\end{coqlisting}

\subparagraph{Lawful monads.}
For brevity, we do not define Functor, Applicative and Monad separately.
We define a class |Monad| and a class |LawfulMonad| (\cref{fig:lawfulmonad}).
Monad contains the functions |map|, |pure|, |ap|, and |bind|.
LawfulMonad extends Monad and further contains |idl|, |idr|, |asc|, |apl|, |apr|, |apl|, and |map_map|,
corresponding to the left and right unit law, and the associativity law of the monad,
and the identity and interchange law of the applicative,
the free theorem of pure,
and the composition law of the functor.

\subparagraph{Static semantics.}
From Coq, we use use units (|tt: Unit|), products (|(a,b): A×B|), functions (|(\a, b): A→B|).
Mirroring the data types of the host language,
we define the types for
unit (|𝟙|),
sums (|s ∨ t|),
products (|s ∧ t|),
functions (|s ⇒ t|) and
effects (|𝕄|) in the guest language (\cref{fig:types}).
We define a data type |ef| to label terms with,
as belonging to the source language |src|,
the target language |tgt|, or
the either language |com| (common) (\cref{fig:lang}).
Label denotation |EF m: ef → (Type → Type)| assigns each functor in the host language.
Concretely, the target and common label is assigned the identity effect functor (e.g. no effect),
and the source language is assigned the effect |m| given as an argument.

We define a data type |tm Γ B t|
for the syntax of our guest language (\cref{fig:tm}).
The terms are parametrized by a type denotation |Γ|, a language label |B| and a type |t|.
The common term formers are %
abstraction |Lam e|, application |App e f| and variables |Var v| to represent functions;
unit |Unt e|, tuple |Prd (e, f)| and projections |Fst e| and |Snd e| to represent products.
The source language has an additional term former |Each e|,
which represents the direct-style effect application |↓| from above.
The target language has additional term formers |Pure e|, |Ap e|, |Map f e|, and |Join e| representing the effect combinators.

The term former |Lam| binds variables. In PHOAS, guest-level bindings are
represented using the host language's bindings. This is why this constructs
takes as an argument a function which binds a variable, represented as a value of type |Γ t|.

\input{fig-types-and-terms.tex}

\subparagraph{Example.}
In our encoding,
the terms of the guest language can be written similarly to the terms of the host language
where each term former of the host language
is wrapped by a term former of the guest language.

For example, the identity function |(λ x, x)| can be encoded in the guest language as |Lam (λ x, Var x)|.
A term |(λ x y, add x y) a b| -- an eta-expanded addition function applied to some arguments --
can be expressed as |Lam (λ x, (Lam (λ y, add `App` (Var y) `App` (Var x)))) `App` (Var a) `App` (Var b)|,
writing application |x `App` y| infix for convenience.
Constructing a term of unit type |tt| is written in the guest language as |Unt tt|.
Similarly, projection on a pair |(a,b).1| is written in the guest language as |Fst (Prd (a,b))|.

\subparagraph{Dynamic semantics.}

Next, we define the dynamic semantics corresponding to the static semantics.
The static semantics has three parts: the types, the language labels and the terms.
Therefore, the dynamic semantics also defines three parts.

The denotation of a type |EVAL m: ty → Type| (\cref{fig:eval-type})
maps each guest type to its corresponding host type,
and is parametrized by a type constructor |m|, corresponding to the monad we evaluate in.

The denotation of a term with regard to the
previously defined type denotation |eval ... : tm (EVAL m) B t → EF m B (EVAL m t))|
(\cref{fig:tm-eval}) interprets the terms in a specific monad depending on which language the terms are labeled from.
More concretely it takes a term of type |t| and of label |B| to be evaluated in monad |m|,
and returns a value of the denotation of the type |EVAL m t| wrapped in the denotation of the label |EF m B|.
The evaluation for terms from the source language implicitly have effects
and can therefore only be interpreted in a monadic interpreter.
For the common and the target language,
we define an evaluation as simply the mapping of guest term formers
to their corresponding host expressions, while mapping variables to variables.

The decision of which monad to use is governed by the label denotation
|EF m: ef → ty → ty| (\cref{fig:eval-lang})
mapping the target and the common language (whose terms do not have implicitly any effects)
to the identity effect, e.g., no effect,
while the |src| language is mapped to the effect |M|.

Note the way we defined the common, source and target terms,
we can relable common terms into source or target terms, e.g.,
into any other language label |relabel: T Γ com t → T Γ e t|.

\subparagraph{Example.}
Consider the evaluation of the following term, which constructs and then destructs a pair of units, which is equal to unit:
|eval com (Fst (Prd (Unt tt, Unt tt))) = tt|.

\input{fig-translation.tex}

\subparagraph{Translation.}
The compilation from the target into the source language is performed by the smart constructor |PURE|,
i.e., we compile from an effectful language into a pure language that uses monadic effect combinators.
We formally define |PURE| (\cref{fig:translation})
that performs both the action of a normal |pure|, e.g.,
wraps the argument into an additional effect |tm Γ src t → tm Γ tgt (M t)|,
and additionally performs a translation from terms form the source language with effect application |Each|
to terms of the target language using combinators |Pure|, |Map|, |Ap| and |Bind|.
This translation makes use of the smart constructors |AP| and |JOIN|,
that perform optimisations.

The |Var|, |App| and |Each| cases were discussed in \cref{sec:translation}:
The direct-style use of effects |Each| is expanded into effect operation |Bind|,
while variables are wrapped in |PURE|, and
function application is translated to applicative |Ap|.
The lambda and empty terms describe values and are simply wrapped into a pure as well.

In the case of projections and the case of tuples,
we follow the general pattern of the homomorphism law, e.g.,
we map both the function (projection, tuple) into a |Pure|
and we wrap all arguments in a |PURE|, and we apply them applicatively.

\subparagraph{Example.}
Assume our language contains an effectful operation |fetch|.
Then, translating the term
|e := Prd (Each (fetch "foo"), Each (fetch "bar"))| yields
|PURE e = Pure (Lam (λ e', (Lam (λ f', Prd (Var e', Var f'))))) `AP` (fetch "foo") `AP` (fetch "bar")|.

\input{fig-span-and-work.tex}

\subparagraph{Span and Work.}

We define span and work~(\cref{fig:span-work}),
which we use to express the degree of parallelism.
Span is the length of the longest chain of unhandled effectful operations,
i.e., the longer the path, the more operations need to run sequentially.
Hence, a shorter span for the same number of operations means a higher amount of parallelism.
Work is the sum of all unhandled effectful operations.
Just like evaluation interprets the value of a term, span and work are interpretations to a numeric value of a term.

As our syntax is defined from types, label and terms, we define these new interpretations as a type denotation, an effect denotation and a term denotation as well.
The effect denotation for span and work is the identity function,
and the type denotation is the constant function mapping all guest types to the type of natural numbers (|SPAN|, |WORK|).

More formally, we define the span of an expression to be zero for variables and values, such as empty and lambda, and for pure expressions.
The span of |Join| and direct-style effect application |Each| is one more (successor |S|) than the span of their argument.
For assertion and projection (access to first and second component),
the span is simply the span of its argument,
while the span of a tuple is the maximum of its left or right branch.
The span for function application, applicative application and mapping is the maximum of the span of its arguments as well, plus the span of the execution of the specified function on the argument.
However, we defined our static semantics such that direct-style effect application cannot be performed under a lambda, therefore the span of the execution of any function is zero.

Analogously,
we define the work of an expression to be zero for variables, values, and pure expressions.
Similar to the span, the work of |Join| and direct-style effect application |Each| is one more (successor |S|) than the work of their argument.
The work of assertion and projection is the work of its argument.
Other than span (which takes the maximum),
the work of a tuple is the sum of both arguments.
The work for function application, applicative application and mapping is the sum of the work of its arguments, plus the work of the execution of the specified function on the argument, which is zero, because lambdas cannot contain |Join| or |Each|.

\subparagraph{Example.}
Assume our language contains an effectful operation |fetch|.
We calculate the span and work of a term in the source language |e := Prd (Each (fetch "foo"), Each (fetch "bar"))| as follows:
|span e = 1| and |work e = 2|.
This expresses the fact that the two effects can be performed in parallel.
The corresponding target language term is\\
|e' := Pure (Lam (λ e', (Lam (λ f', Prd (Var e', Var f'))))) `AP` (fetch "foo") `AP` (fetch "bar")|.\\
We get the same results for this term:
|span e' = 1| and
|work e' = 2|.

\subsection{Proof}

Our translation should only change the encoding from direct-style to effect combinators, while the semantics, typability and parallelism of the term should be preserved.
We prove that our translation preserves typability, semantics, span and work.
Intuitively, the theorems hold,
because our translation performed by |PURE|, |AP|, and |JOIN|
are the functor, monad and applicative laws.

\begin{theorem}[PURE preserves types] \normalfont
  The translation function takes a well-typed term and produces a well-typed term, i.e.,
  \lstinline|PURE: ∀ t, tm Γ src t → tm Γ tgt (𝕄 t)|
\end{theorem}
\begin{proof}
Using intrinsically-typed representation of terms,
the well-typedness of the translated term is guaranteed by the fact that
the definition of the translation function |PURE| is itself well-typed in Coq.
\end{proof}

We now consider the preservation of semantics.
First, we show that the semantics of the smart constructors is equal to that of the normal constructors,
so that they merely represent optimizations of those.
\begin{lemma}[AP respects semantics]\label{lem:eval-ap} \normalfont
  \lstinline|∀ f e, eval tgt (AP f e) = eval tgt (Ap f e)|
\end{lemma}
\begin{lemma}[JOIN respects semantics]\label{lem:eval-join} \normalfont
  \lstinline|∀ f e, eval tgt (JOIN e) = eval tgt (Join e)|
\end{lemma}
\begin{proof}
By case distinction on the term structure of the arguments, using the functor, monad and applicative laws.
\end{proof}

Next, we see that embedding the pure |com| sublanguage in the target language preserves
the semantics:
\begin{lemma}[relabel preserves semantics]\label{lem:eval-relabel} \normalfont
  \lstinline|∀ e, eval tgt (relabel e) = eval com e|
\end{lemma}
\begin{proof}
By induction on the structure of |e|.
\end{proof}

From this, we can deduce that the |PURE| transformation preserves the semantics of the source
program.
\begin{theorem}[PURE preserves semantics]\label{lem:eval-pure}
  For all lawful monads |M| to be evaluated in,\\
  \lstinline|∀ e, eval tgt (PURE e) = eval src e|
\end{theorem}
\begin{proof}
By induction on the structure of |e|, using \cref{lem:eval-ap,lem:eval-join,lem:eval-relabel}.
\end{proof}

We now want to show that |PURE| preserves the work and span of the program.
This is similar to semantics preservation, except that the functions we consider
map to a monoid (the natural numbers with addition and maximum, respectively)
rather than a monad.

We show that |AP| and |JOIN| do not increase the span and work of a term, compared to
the normal constructors.
\begin{lemma}[AP respects span and work]\label{lem:span-work-ap} \normalfont \quad \\
|∀ f e,| ~~ |span tgt (AP f e) <= span tgt (Ap f e)| ~~and~~ |work tgt (AP f e) <= work tgt (Ap f e)|
\end{lemma}
\begin{lemma}[JOIN respects span and work]\label{lem:span-work-join} \normalfont \quad \\
|∀ e,| ~~ 
|span tgt (JOIN e) <= span tgt (Join e)| ~~and~~ |work tgt (JOIN e) <= work tgt (Join e)|
\end{lemma}
\begin{proof}
By case distinction on the term structure of the arguments, using the monoid laws.
\end{proof}

The pure terms in the |com| sublanguage are effect-free; therefore, their span and work is equal to 0.
\begin{lemma}[com is effect-free]\label{lem:com-free} \normalfont
|∀ e, span com e = 0| ~~and~~ |work com e = 0|
\end{lemma}
\begin{proof}
By induction on the term structure of |e|.
\end{proof}

Embedding pure terms into the targe language produces a term that does not perform any effects, either.
\begin{lemma}[relabeled terms remain effect-free]\label{lem:relabel-free} \normalfont \quad \\
|∀ e,| ~~ 
|span tgt (relabel e) = 0| ~~and~~ |work tgt (relabel e) = 0|
\end{lemma}
\begin{proof}
By induction on the term structure of |e|.
\end{proof}

We can then show that the translation |PURE| does not increase the span or work of the 
source program, thereby demonstrating that it is parallelism-preserving.
\begin{theorem}[PURE preserves span and work]\label{lem:span-work-pure} \normalfont \quad \\
|∀ e,| ~~ 
|span tgt (PURE e) ≤ span src e| ~~and~~ |work tgt (PURE e) ≤ work src e|
\end{theorem}
\begin{proof}
By induction on the term structure of |e|, using the monoid laws of addition and maximum
as well as \cref{lem:span-work-ap,lem:span-work-join,lem:com-free,lem:relabel-free}.
\end{proof}

%% file: fig-types-and-terms.tex
\begin{figure}
\captionof{lstlisting}{Syntax and semantics.}
\begin{sublstlisting}{.28\textwidth}
\caption{Types.}\label{fig:types}
\begin{coqlisting}
Inductive
ty: Type :=
| 𝟙: ty
| ∨: ty → ty → ty
| ∧: ty → ty → ty
| ⇒: ty → ty → ty
| 𝕄: ty → ty.
\end{coqlisting}
\end{sublstlisting} \begin{sublstlisting}{.41\textwidth}
\caption{Type denotation.}\label{fig:eval-type}
\begin{coqlisting}
Equations
EVAL (m: Type → Type): ty → Type :=
| m,   𝟙    => Unit
| m, s ∨ t  => EVAL m s + EVAL m t
| m, s ∧ t  => EVAL m s × EVAL m t
| m, s ⇒ t  => EVAL m s → EVAL m t
| m, 𝕄 t   => m (EVAL m t).
\end{coqlisting}
\end{sublstlisting} \begin{sublstlisting}{.29\textwidth}
\caption{Labels and denotation.}\label{fig:lang}\label{fig:eval-lang}
\begin{coqlisting}
Inductive
ef := src | tgt | com.
Equations
EF m: ef → Type → Type :=
| m,src, t => m t
| m,com, t => t
| m,tgt, t => t.
\end{coqlisting}
\end{sublstlisting}

\vspace{1em}

\begin{sublstlisting}{\textwidth}
\caption{Term and their denotation.}\label{fig:tm}\label{fig:tm-eval}
\begin{coqlisting}
Inductive tm {Γ: ty → Type}: ef → ty → Type :=
| Var {B t}:   Γ t              → tm B t
| Unt {B}:     Unit             → tm B 𝟙
| Prd {B s t}: tm B s × tm B t  → tm B (s ∧ t) 
| Fst {B s t}: tm B (s ∧ t)     → tm B s
| Snd {B s t}: tm B (s ∧ t)     → tm B t
| App {B s t}: tm B (s ⇒ t)     → (tm B s → tm B t)
| Lam {B s t}: (Γ s → tm com t) → tm B (s ⇒ t)

| Each {t}:    tm src (𝕄 t)    → tm src t

| Pure {t}:    tm com t            → tm tgt (𝕄 t)
| Join {t}:    tm tgt (𝕄 (𝕄 t))  → tm tgt (𝕄 t)
| Map  {s t}:  tm tgt (s ⇒ t)      → (tm tgt (𝕄 s) → tm tgt (𝕄 t))
| Ap  {s t}:   tm tgt (𝕄 (s ⇒ t)) → (tm tgt (𝕄 s) → tm tgt (𝕄 t)).

Equations eval {t m} {M:Monad m} B: tm (EVAL m) B t → EF m B (EVAL m t) :=
| src, Var i       => M.(pure) i                              (* src *)
| src, Lam k       => M.(pure) (eval · k)
| src, Unt tt      => M.(pure) tt
| src, Fst e       => M.(map) (λ e', e'.1) (eval e)
| src, Snd e       => M.(map) (λ e', e'.2) (eval e)
| src, App e f     => M.(ap) (eval e) (eval f)
| src, Prd (e, f)  => M.(ap) (M.(map) (λ a' b', (a', b')) (eval e)) (eval f)
| src, Each e      => M.(bind) id (eval e)

| _,   Var i       => i                                       (* com or tgt *)
| _,   Lam k       => eval · k
| _,   Fst e       => (eval e).1
| _,   Snd e       => (eval e).2
| _,   App e f     => (eval e) (eval f)
| _,   Prd (e, f)  => (eval e, eval f)
| _,   Unt tt      => tt
| tgt, Map f e     => M.(map)  (eval f) (eval e)              (* only tgt *)
| tgt, Ap f e      => M.(ap)   (eval f) (eval e)
| tgt, Pure e      => M.(pure) (eval e)
| tgt, Join e      => M.(bind) id (eval e).
\end{coqlisting}

\end{sublstlisting}
\end{figure}

%% file: fig-translation.tex
\begin{coqlisting}[float, caption={Translation.}, label={fig:translation}]
Notation "f `AP` e" := (AP f e) (at level 20).

Equations PURE {Γ x} (e: tm Γ src x): tm Γ tgt (𝕄 x) :=
| Var i       => Pure (Var i)
| Unt tt      => Pure (Unt tt)
| Lam j       => Pure (Lam j)
| Fst e       => Pure (Λ e', Fst (Var e')) `AP` PURE e
| Snd e       => Pure (Λ e', Snd (Var e')) `AP` PURE e
| Prd (e, f)  => Pure (Λ e' f', Prd (Var e',  Var f')) `AP` PURE e `AP` PURE f
| App e f     => PURE e `AP` PURE f
| Each e      => JOIN (PURE e).

Equations AP {Γ s t} (f: tm Γ tgt (𝕄 (s ⇒ t))) (e: tm Γ tgt (𝕄 s)): tm Γ tgt (𝕄 t) :=
| Pure f, Pure e => Pure (App f e)
| Pure f,      e => Map (Lam (λ x, App f (Var x))) e
|      f, Pure e => Map (Lam (λ x, App (Var x) e)) f
|      f,      e => Ap f e.

Equations JOIN {Γ t} (e: tm Γ tgt (𝕄 (𝕄 t))): tm Γ tgt (𝕄 t) :=
| Pure e => to e
|      e => Join e.
\end{coqlisting}

%% file: fig-span-and-work.tex
\begin{figure}
\captionof{lstlisting}{Span and Work.}\label{fig:span-work}
\begin{sublstlisting}{.5\linewidth}
\caption{Span.}\label{fig:span}
\begin{coqlisting}
Equations SPAN: ty → Type := | _ => nat.
Equations
span {B x} (e: tm SPAN B x): nat :=
| Var i  => 0 | Lam e  => 0
| Unt tt => 0 | Pure e => 0
| Fst  e     => span e
| Snd  e     => span e
| Prd  (e,f) => max (span e) (span f)
| App  e f   => max (span e) (span f)
| Ap   e f   => max (span e) (span f)
| Map  e f   => max (span e) (span f)
| Join e     => S (span e)
| Each e     => S (span e).
\end{coqlisting}
\end{sublstlisting}
\begin{sublstlisting}{.5\linewidth}
\caption{Work.}\label{fig:work}
\begin{coqlisting}
Equations WORK: ty → Type := | _ => nat.
Equations
work {B x} (e: tm WORK B x): nat :=
| Var i  => 0 | Lam e  => 0
| Unt tt => 0 | Pure e => 0
| Fst  e       => work e
| Snd  e       => work e
| Prd  (e,f)   => work e + work f
| App  e f     => work e + work f
| Ap  e f      => work e + work f
| Map  e f     => work e + work f
| Join e       => S (work e)
| Each e       => S (work e).
\end{coqlisting}
\end{sublstlisting}
\end{figure}

%% file: implementation.tex
\section{Implementation}
\label{sec:implementation}

In this section, we describe the differences and similarities
between the mechanisation in Coq and the implementation in Scala built on a macro-based AST transformation.

\subparagraph{Structural Recursion.}
We keep our implementation in a general-purpose language as close to the formal model of our core calculus as possible.
To this end, our implementation follows the formal translation as a structurally recursive function over the terms where possible.
We use Scala macros to get access to code as AST data type, similar to the |tm| data type in the formalization.

\subparagraph{Type-preserving Compilation.}
In Scala, we process the untyped AST for fine-grained detailed manipulations.
Knowing that the translation is typability-preserving by our Coq proof,
increases confidence in the implementation.

\subparagraph{Exhaustiveness Checks.}
A difference between the Coq and the Scala implementation is that
|Bind|, |Ap| and |Pure| are not syntax forms in Scala
but represented by variable and function application in the embedding.
Still, we can treat them as syntax forms to construct and destruct
by defining custom patterns for pattern matching.
Further, Scala macros do not define the Scala syntax as an algebraic data type
(to hide compiler internals),
and therefore do not offer exhaustiveness checks.
Yet, by the fact that the Coq implementation is total, the Scala implementation can be expected to be as well.

\subparagraph{Custom Effects.}
In the formalization,
we have only a single effect,
while, in the implementation, we allow every use of the notation
to be instantiated with a different effect, based on the type of the expression.
Our macro inspects the expression's type and,
based on this type, picks the corresponding generated combinators |bind|/|ap|/|pure|
of the respective effect.

\subparagraph{Arity.}
In Scala, functions may take multiple arguments.
Generally, we can model functions taking multiple arguments as functions
taking a single argument of a tuple with multiple fields
with appropriate currying and uncurrying.
Functions with multiple arguments make the compiler no longer structurally recursive over terms
because, besides terms, the compiler additionally needs to mutually recurse over the list of arguments, 
which would complicate the proof, but is necessary for our implementation.

%% file: related.tex
\section{Related work}
\label{sec:related-work}

\subparagraph{Do-Notation.}
Do-notations have been popular for studying a variety of styles
for writing effectful code:
Wadler extends list-comprehension syntax~\cite{Wadler92}
to monadic comprehensions, from which modern do-notation sprung, and
McBride introduced applicatives and idiom brackets as a notation for applicatives~\cite{McbrideP08}.
To the best of our knowledge,
the only support for mixed sequential and parallel programming
was introduced as a Haskell extension~\citet{MarlowBCP14, MarlowJKM16}
to optimise do-notation into mixed monadic/applicative operations (|ApplicativeDo|).
In contrast, our notation preserves the parallelism inherent in the structure of the program, thereby allowing sequentiality where necessary
and giving parallelism where possible.

\subparagraph{Implementations.}
Besides theory,
implementations for effectful guest language notations are a popular endeavor,
for example:
In Scala, we can find projects to supports effectful programs
through compiler plugins such as
coroutines \cite{scala-coroutines},
Scala async \cite{scala-async},
Monadless~\cite{monadless},
Effectfull~\cite{effectful},
Scala Workflow~\cite{scala-workflow},
Scala ContextWorkflow~\cite{InoueAI18},
Scala Computation Expressions~\cite{scala-computation-expressions},
Dsl.scala~\cite{dslscala},
Dotty CPS~\cite{dotty-cps-async}.
In other languages we have:
F\# computation expressions~\cite{PetricekS14},
In particular proof-assistants and dependently typed languages
have an interest for good support of notations for guest languages,
which we can see in Idris’~\cite{Brady14, idris-bang-notation}
Lean’s~\cite{UllrichM22, lean-do-notation},
and Kind’s~\cite{kind-language} notation.
None of them support parallelism.

Further, the following approaches are similar to |ApplicativeDo|:
OCaml’s monadic and applicative let~\cite{ocaml-applicative-let},
Scala avocADO \cite{scala-avocado}, and
Scala parallel-for \cite{scala-parallel-for}.
But these do not support direct-style effect usage, and do not preserve parallelism.

\subparagraph{CPS Translations.}
In general, effects are implemented
by translating to other already known effects. In particular, all effects can be
represented by the continuation effect~\cite{Filinski94}, and thus,
by translating to continuation passing style (CPS)~\cite{Reynolds72, Fischer72}.
However, naive CPS translations introduce so called administrative redexes, e.g.,
expressions containing subexpressions which do not need to be evaluated at run-time,
but can already be optimised by a partial evaluation pass at compile-time.
Eventually, Danvy and Nielsen~\cite{DanvyN03} optimised the CPS-translation
into a first-order, one-pass, compositional translation.

Their trick for achieving an optimal result in one pass is to build optimisations
into the definitions of their translation functions.
We use a similar approach in our translation through the definition of smart constructors
which simplify terms using monad and applicative laws when called.

\subparagraph{Host supporting effects.}
Because effects can be implemented by translation to equally or more powerful effects,
besides giving a denotational semantics modelling a compiler,
there is another approach – that we did not follow –
by forwarding effects to the host language as well.
Then, compile-time translations like ours can be avoided and effects can be implemented in languages as a library, given the host languages has sufficient powerful effects.
Filinski~\cite{Filinski94} studied the implementation of effects in languages with delimited continuations (e.g. |shift| and |reset|).
In such an impure language it is possible to implement so called monadic reflection – a
function taking an effectful function and returning a ``pure'' function.
This is of course only possible
by exploiting the impurity of the host language to implement the effect using delimited continuation.
Later, Forster~\cite{ForsterKLP19} studied the translations between monadic reflection, effect handlers and delimited control.
The approach to extend the underlying virtual machine by support for delimited continuations, which are sufficiently efficient for then implementing effects
as normal libraries is followed by:
the JVM proposal for delimited continuations~\cite{project-loom},
the Haskell proposal for continuation marks~\cite{haskell-cps} and
multicore-ocaml~\cite{multicore-ocaml}.

We are looking for a more general solution for compiling a language,
that works independent of whether the runtime already supports delimited continuations or not.

\subparagraph{Formalisation Techniques.}

To focus on the interesting parts of our formalisation,
we used modern techniques to define features of the guest language
in terms of features of the host language:
In particular, we use
parametric higher order abstract syntax (PHOAS)~\citet{PfenningE88, Hofmann99, Chlipala08}
to inherit binders and capture-avoiding substitution from the host language,
and intrinsically-typed syntax~\citet{Danielsson06, BentonHKM12, AltenkirchCU14, AltenkirchCU14b} to inherit type checking.
The choice of PHOAS implies a limitation of our work,
namely that we can formally only prove theorems about closed terms.
Yet, this is a common restriction and lifting it is subject to future work.

%% file: conclusion.tex
\section{Conclusion}
\label{sec:conclusion}

Existing notations for composing effectful computations fall short on
providing both sequential and parallel composition of effects at the same time.
In this paper, we proposed a notation for mixed sequential/parallel code.
Our notation allows direct-style effects,
a feature that enables the sequentiality or parallelism of the effects to be determined by the structure of the code.
We proved that our compilation preserves the parallelism of the source program and mechanized the proof in Coq.

An interesting next step for this line of research on direct-style notations for effects
is to investigate how to cover
more programming language features such as loops and branches,
to integrate effects more seamlessly into the language.
Besides monad and applicative functors,
other effect functors,
such as selectives~\cite{MokhovLMD19, WillisWP20}, comonads~\cite{OrchardM12},
and the theory behind effectful recursion~\cite{ErkokL02}
and generalizations such as arrows~\cite{Hughes00, LindleyWY11}
are promising possibilities.